\documentclass[onecolumn, draftcls,11pt]{IEEEtran}

\usepackage{amsfonts}                                                          
\usepackage{amsthm}
\usepackage{graphicx}        
\usepackage{latexsym}
\usepackage{amssymb}
\usepackage{amsmath}
\usepackage{array}
\usepackage{balance}
\usepackage{booktabs}
\usepackage[font=small]{caption}
\usepackage{cite}
\usepackage{epstopdf}
\usepackage{epsfig}
\usepackage{etoolbox}
\usepackage[bottom]{footmisc}
\usepackage[usenames,dvipsnames]{xcolor}
\usepackage{hyperref}
\usepackage{mathptmx}
\usepackage{mathtools}
\usepackage{multirow}
\usepackage{parskip}
\usepackage{tikz}
\usepackage{soul}
\usepackage{verbatim}
\usepackage{calrsfs}

\usetikzlibrary{decorations.pathreplacing}

\allowdisplaybreaks

\hypersetup{
    colorlinks=true,
    linkcolor=black,
    filecolor=magenta,      
    urlcolor=blue,
    citecolor=Orange,
}

\DeclareMathAlphabet{\pazocal}{OMS}{zplm}{m}{n}

\let\bbordermatrix\bordermatrix
\patchcmd{\bbordermatrix}{8.75}{4.75}{}{}
\patchcmd{\bbordermatrix}{\left(}{\left[}{}{}
\patchcmd{\bbordermatrix}{\right)}{\right]}{}{}

\newcommand{\sr}{\stackrel}

\newcommand{\rar}{\rightarrow}

\newcommand{\tri}{\sr{\triangle}{=}}

\newcommand{\be}{\begin{equation}}
\newcommand{\ee}{\end{equation}}
\newcommand{\bea}{\begin{eqnarray}}
\newcommand{\eea}{\end{eqnarray}}
\newcommand{\bes}{\begin{eqnarray*}}
\newcommand{\ees}{\end{eqnarray*}}
\newcommand{\bce}{\begin{center}}
\newcommand{\ece}{\end{center}}
\newcommand{\beae}{\begin{IEEEeqnarray}{rCl}}
\newcommand{\eeae}{\end{IEEEeqnarray}}

\def\VR{\kern-\arraycolsep\strut\vrule &\kern-\arraycolsep}
\def\vr{\kern-\arraycolsep & \kern-\arraycolsep}

\newcommand{\ben}{\begin{enumerate}}
\newcommand{\een}{\end{enumerate}}

\newcommand{\hso}{\hspace{.1in}}
\newcommand{\hst}{\hspace{.2in}}

\newtheorem{theorem}{Theorem}[section]

\newtheorem{remark}{Remark}[section]
\newtheorem{corollary}{Corollary}[section]

\newtheorem{assumptions}{Assumptions}[section]

\newtheorem{proposition}{Proposition}[section]

\newtheorem{observation}{Observation}[section]

\begin{document}

\title{On the Fragile Rates of  Linear Feedback Coding Schemes of Gaussian  Channels with Memory}

 \author{
   \IEEEauthorblockN{
     Charalambos D. Charalambous\IEEEauthorrefmark{1} and   
    Christos Kourtellaris\IEEEauthorrefmark{1} and 
   Themistoklis Charalambous\IEEEauthorrefmark{2}
     \\}
   \IEEEauthorblockA{
     \IEEEauthorrefmark{1}Department of Electrical and Computer Engineering\\University of Cyprus, Cyprus\\ 
  }
  \IEEEauthorblockA{
    \IEEEauthorrefmark{2}Department of Electrical Engineering and Automation, School of Electrical Engineering\\ Aalto University,  Finland}\\
    Emails: chadcha@ucy.ac.cy, kourtellaris.christos@ucy.ac.cy,
     themistoklis.charalambous@aalto.fi
 }

\maketitle

%
%
%
%
\begin{abstract}
In \cite{butman1976} the  linear coding scheme is applied,    $X_t =g_t\Big(\Theta - {\bf E}\Big\{\Theta\Big|Y^{t-1}, V_0=v_0\Big\}\Big)$, $t=2,\ldots,n$,  $X_1=g_1\Theta$, with $\Theta: \Omega \rar {\mathbb R}$, a   Gaussian random variable, to  derive  a lower bound on the feedback rate, for    additive Gaussian noise (AGN) channels, $Y_t=X_t+V_t, t=1, \ldots, n$, where $V_t$ is a  Gaussian  autoregressive (AR)  noise, and  $\kappa \in [0,\infty)$ is the total transmitter power.  For the unit memory AR noise, with parameters $(c, K_W)$, where $c\in [-1,1]$ is the pole and $K_W$ is the variance of the Gaussian noise,  the lower bound is $C^{L,B} =\frac{1}{2} \log \chi^2$, where   $\chi =\lim_{n\longrightarrow \infty} \chi_n$ is the positive root of $\chi^2=1+\Big(1+ \frac{|c|}{\chi}\Big)^2 \frac{\kappa}{K_W}$, and the sequence $\chi_n \tri \Big|\frac{g_n}{g_{n-1}}\Big|, n=2, 3, \ldots,$ satisfies a certain  recursion, and   conjectured that $C^{L,B}$ is the feedback capacity.  The conjectured is proved  in \cite{kim2010}. 

In this correspondence, it is  observed that the  nontrivial  lower bound $C^{L,B}=\frac{1}{2} \log \chi^2$  such that   $\chi >1$,  necessarily implies  the  scaling coefficients of the feedback  code,   $g_n$, $n=1,2, \ldots$, grow  unbounded, in the sense that, $\lim_{n\longrightarrow\infty}|g_n| =+\infty$. The unbounded behaviour of $g_n$ follows from the ratio limit theorem of a sequence of real numbers, and   it is verified  by simulations. It is then concluded that such  linear codes are not practical, and  fragile with respect  to  a mismatch between the statistics of the mathematical model of the channel and the real statistics of the channel. In particular, if the error is perturbed by $\epsilon_n>0$ no matter  how small, then $X_n =g_t\Big(\Theta - {\bf E}\Big\{\Theta\Big|Y^{t-1}, V_0=v_0\Big\}\Big)+g_n \epsilon_n$, and $|g_n|\epsilon_n \longrightarrow \infty$, as $n \longrightarrow \infty$. 
\end{abstract}

%
%
%
%
 \section{Introduction, Main Results,  Literature  Review and Observations}
\label{intro}
Achievable lower and upper bounds on feedback rates of additive Gaussian (AGN) channels with memory, driven by autoregressive AR noise,  are derived   in the   early 1970's,  in \cite{butman1969,tienan-schalkwijk1974,wolfowitz1975,butman1976},    using generalizations of Elias  \cite{elias1961}, and Schalkwijk and Kailath \cite{schalkwijk-kailath1966},   coding schemes of memoryless AGN channel.  Bounds are also derived in \cite{dembo1989,ozarow1990}, and compared to Butman's bounds \cite{butman1976}. Variations of the coding schemes   \cite{elias1961,schalkwijk-kailath1966}, are applied to memoryless AGN channel with feedback,  in the context of joint source channel-coding, using posterior matching feedback schemes  in \cite{shayevitz-feder2007,shayevitz-feder2008,shayevitz-feder2009}. \\
In  \cite{yang-kavcic-tatikonda2007},  the  ``maximal information rate'' of the AGN channel with unit memory stationary AR noise is computed   (see Corollary 7.1 in \cite{yang-kavcic-tatikonda2007}), and noted  it is identical to Butman's lower bound. In  \cite{kim2010},  Butman's lower bound is shown to correspond to the feedback capacity of the AGN channel with unit memory stationary AR noise, while  additional generalizations are also obtained for  stationary autoregressive moving average noise. 


In this paper, we identify fundamental {\it fragile} properties of the  linear feedback coding scheme applied in \cite{butman1976} to derive  the   lower bound on feedback capacity.  To keep our analysis and observations as simple as possible, our  discussion of \cite{butman1976}   is focused on AGN channels driven by the simplest noise with memory, the autoregressive AR unit-memory Gaussian noise.    However, our observations are not limited by the simplicity of the noise. 

\subsection{Additive Gaussian Noise Channels Driven by  Autoregressive Noise}
Bounds on the feedback capacity of  AGN channels are derived\footnote{In \cite{tienan-schalkwijk1974,butman1976,wolfowitz1975} the noise is  time-invariant, stable or marginally stable.} in  Tienan's and Schalkwijk's 1974 paper \cite{tienan-schalkwijk1974}, Wolfowitz's 1975 paper \cite{wolfowitz1975}  and Butman's 1976 paper \cite{butman1976}, {where the authors presuppose  the initial state of the noise  is} known to the encoder and the decoder\footnote{ \cite[page 311 below the noise model]{tienan-schalkwijk1974}, where $z_{-m+1},\ldots, z_0$ is used, and rest of pages where rates are conditioned on ${\cal Z}^m$;  \cite[Section~I, second pagagraph]{wolfowitz1975},   ``$z_0$ is the state of the channel at the beginning of transmission; $z_0$  is known to both sender and receiver'';    \cite[eqn (17)]{butman1976}, where $n_0$ is used.}. 

Below, we introduce the AGN channel driven by time-varying AR noise,  with respect to Butman's \cite{butman1976}    linear time-varying feedback coding scheme, as shown in  Figure~\ref{fig:system}. This generalization is considered to keep our  presentation more interesting, and to verify via an alternative derivation,   that Butman's lower bound on achievable feedback rate, also  holds for the more general nonstationary and nonergodic AGN channels investigated by  Cover and Pombra in \cite{cover-pombra1989} (even though we show the coefficients of the error of  coding  scheme grow unbounded).    \\
Butman in \cite{butman1976} considers the restriction $c_t=c \in [-1,1], K_{W_t}=K_W\in (0,\infty),\forall t $ (i.e.,  includes nonasymptotically stationary noise),  while  Tienan and Schalkwijk's, and also Wolfowitz    \cite{tienan-schalkwijk1974,wolfowitz1975} consider  the restriction
$c_t=c \in (-1,1), K_{W_t}=K_W\in (0,\infty),\forall t $ (i.e., asymptotically stationary noise).

{\it AGN Driven by Time-Varying  AR Noise AR$(c_t;v_0)$.}
\begin{align}
&Y_t=X_t+V_t,  \hst t=1, \ldots, n, \label{b_ar_1} \\
&\frac{1}{n} {\bf E} \Big\{\sum_{t=1}^n\big(X_t\big)^2\Big|V_0=v_0\Big\} \leq \kappa, \hst  \kappa \in [0,\infty), \label{b_ar_2} \\
&V_t= c_t V_{t-1}+ W_t, \hso V_0=v_0,   \hso c_t \in (-\infty, \infty),  \hso t=1, \ldots, n,  \label{b_ar_3}\\
&W_t\in G(0,K_{W_t}),\: K_{W_t}\in (0,\infty),  \hso  t=1, \ldots, n,    \hso \mbox{indep. Gaussian seq., indep. of $V_0\in G(0, K_{V_0})$}, \label{b_ar_4} \\
& X_t=e_t(v_0,
\Theta, Y^{t-1})), \hso \mbox{$e_t(\cdot)$ is linear  in $(v_0,\Theta, Y^{t-1})$},\hso  t=2, \ldots, n, \label{b_ar_6}\\
&X_1= e_1(v_0,\Theta),  \hso \mbox{$e_1(\cdot)$ is linear  in $(v_0,\Theta)$}, \label{b_ar_7}\\
& \Theta : \Omega \rar {\mathbb R} \hso \mbox{is a Gaussian message},  \hso   \Theta \in G(0,K_\Theta), \: K_\Theta >0,\\
&W_t, \hso t=1, \ldots, n, \hso V_0, \hso \Theta \hso \mbox{are mutually independent}    \label{b_ar_5_a}
\end{align}
where the RVs $X_t$, $Y_t$ and $V_t$ are defined  as follows. \\
$X^{n} \tri \{X_1, X_2, \ldots, X_n\}$ is  the sequence of channel input random variables (RVs) $X_t :  \Omega \rar {\mathbb R}$,   \\
$Y^{n} \tri \{Y_1, Y_2, \ldots, Y_n \}$ is  the sequence of channel output RVs $Y_t :  \Omega \rar {\mathbb R}$,\\
   $V^n\tri \{ V_1, \ldots, V_n\}$ is the  sequence of  jointly Gaussian distributed  RVs $V_t :  \Omega \rar {\mathbb R}$, for fixed $V_0=v_0$,  \\
 $V_0\tri v_0$, is the initial state of the channel, i.e., of the noise, known to the encoder and decoder,   \\
 $G(0, K_X)$ denotes a Gaussian  distribution induced by a Gaussian RV, $X: \Omega \rar {\mathbb R}$,  with zero mean  and variance $K_X$. 

\begin{figure}[ht]
    \centering
    \includegraphics[width=0.6\columnwidth]{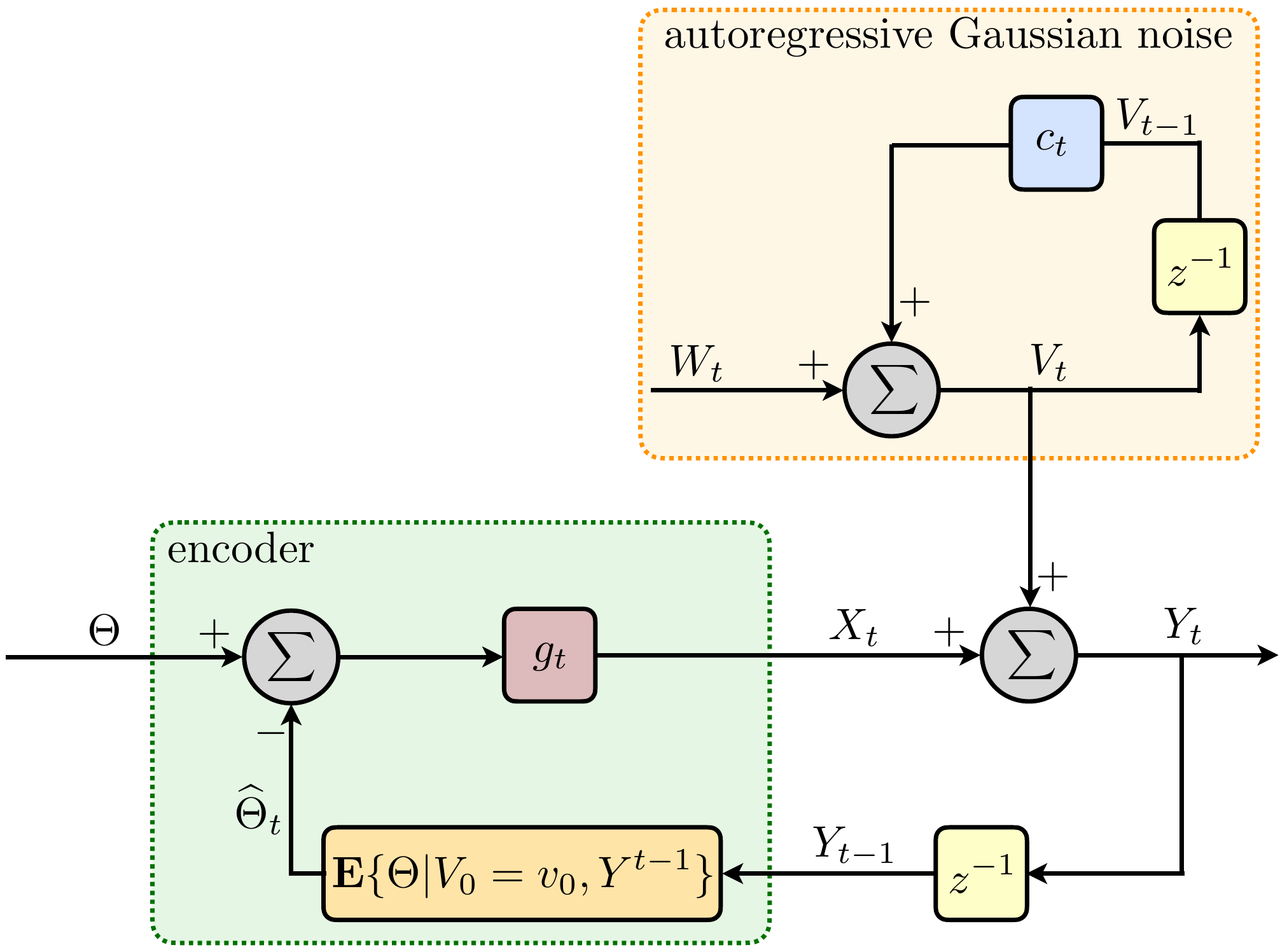}
    \caption{Mathematical model of the AGN channel driven by a time-varying AR$(c_t;v_0), c_t \in (-\infty, \infty), K_{W_t}>0\; \forall t$ noise, with Butman's linear feedback coding scheme. The model in \cite{butman1976} corresponds to the restriction AR$(c;v_0), K_{W_t}=K_W \forall t$. $g_n, n=1,2, \ldots$ is  gain sequence multiplying the  estimation error. }
    \label{fig:system}
\end{figure}

Throughout the paper, we  use the following notation.\\
 AR$(c_t;v_0)$,  denotes the  time-varying autoregressive unit memory noise $V^n$, with  $c_t \in (-\infty, \infty), K_{W_t} \in (0,\infty),  t=1,\ldots, n$,   AR$(c;v_o)$  denotes  its restriction to  time-invariant,  with  $c_t=c \in (-\infty, \infty), K_{W_t}=K_W \in (0,\infty),  t=1,\ldots, n$.  The stable AR$(c;v_0)$ noise corresponds to $c\in (-1,1)$.  \\
{A time-varying AR noise without an initial state is denoted by AR$(c_t), c_t\in (-\infty,\infty)$, and corresponds to the case, the RV $V_0$ generates  the trivial information, i.e.,  the sigma$-$algebra generated by $V_0$ is  $\sigma\{V_0\}=\{\emptyset, \Omega\}$.}\\
$I(\Theta;Y^n|v_0)$  denotes  the mutual information between the Gaussian message $\Theta$ and the sequence $Y^n$ conditioned on $V_0=v_0$, and 
$I(\Theta;Y^n)$ denotes  the mutual information between $\Theta$ and the sequence $Y^n$.\\ From the above formulation,  follows that,  if  $V_0$ generates the trivial information, then $I(\Theta;Y^n|V_0)=I(\Theta;Y^n)$; in this case $I(\Theta;Y^n)$ depends on the distribution of the RV $V_1$, i.e., ${\bf P}_{V_1}$ and $V_1=W_1$.
 
Butman in the 1976 paper \cite{butman1976}, considered (\ref{b_ar_1})-(\ref{b_ar_5_a}), for AR$(c;v_0), c \in [-1,1]$, i.e., $c_t =c \; \forall t$, $K_\Theta=1$, and  derived achievable feedback  rates, based on   the optimization problem 
\begin{align}
C_{n}^L(\kappa, v_0) \tri \sup_{X_t=e_t(v_0, \Theta, Y^{t-1}), \; e_t(\cdot) \: \:  \mbox{is linear}, \:\:    t=1, \ldots, n: \:    \frac{1}{n}  {\bf E} \big\{\sum_{t=1}^{n} \big(X_t\big)^2\Big|V_0=v_0\big\} \leq \kappa} \frac{1}{n}I(\Theta;Y^{n}|v_0)\label{ir} , 
\end{align}
and its per unit time limit, 
\begin{align}
C^L(\kappa, v_0) = \lim_{n \longrightarrow \infty} C_{n}^L(\kappa,v_0)\label{ir_a} 
\end{align}
provided the supremum and limit exist.    In particular,   Butman proved  that linear strategies $e_t(\cdot), t=1, \ldots, n$ are given by \cite[eqn(17)]{butman1976},
\begin{align}
&e_t(v_0, \Theta, Y^{t-1})= g_t\Big(\Theta - {\bf E}\Big\{\Theta\Big|Y^{t-1}, V_0=v_0\Big\}\Big), \hst t=2,\ldots,n, \hst e_1(v_0,\Theta)=g_1\Theta  \label{butman_cs_1_new}
 \end{align}
where $g_1, g_2, \ldots, g_n$ is a sequence of nonrandom real numbers. Hence, the supremum in (\ref{ir}) is replaced by the supremum over the  sequence  $g_1, g_2, \ldots, g_n$ that satisfies the average power constraint.

\begin{remark} The AGN Channel Driven by Noise without Initial State\\
 The variation of the above AGN channel without an  initial state, follows directly from the Tienan and Schalkwijk  \cite{tienan-schalkwijk1974} and    Butman \cite{butman1976} formulation, by letting   $V_0$ generate  the trivial information, i.e.,  the sigma$-$algebra generated by $V_0$ is  $\sigma\{V_0\}=\{\emptyset, \Omega\}$. In such a  formulation  $C_{n}^L(\kappa, v_0), C_{n}^L(\kappa, v_0)$ are replaced by $C_{n}^L(\kappa, {\bf P}_{V_1}), C_{n}^L(\kappa, {\bf P}_{V_1})$, which emphasize their dependence on the distribution ${\bf P}_{V_1}$ of $V_1$, instead of initial state $V_0=v_0$.
\end{remark}

 \cite{butman1976} is focused on  the   optimization problem (\ref{ir})  and its per unit time limit (\ref{ir_a}), with coding scheme  (\ref{butman_cs_1_new}),  and in particular  on the derivation of upper and lower  bounds. { It  will become apparent (in subsequent sections) that Butman's lower bound is derived using a per-symbol average power constraint at each transmission time, i.e.,  ${\bf E} \Big\{\big(X_n\big)^2\Big|V_0=v_0\Big\} \leq \kappa, n=1,2, \ldots$,     and not    $\frac{1}{n}  {\bf E} \big\{\sum_{t=1}^{n} \big(X_t\big)^2\Big|V_0=v_0\big\} \leq \kappa$. } \\Over the years,  the following  result, is  used extensively in the literature, such as, \cite{liu-han2019,li-elia2019}. 

\ben
\item[(R)] Butman's  lower bound on $C_n^L(\kappa, v_0)$ and $C^L(\kappa, v_0)$, and Butman's Conjecture, that these bounds  correspond to the feedback capacity \cite[Abstract]{butman1976}. 
\een
{In particular,  \cite{kim2010} proved   that Butman's Conjecture is correct, and the frequency and time-domain characterizations of feedback capacity\footnote{These theorems are used in \cite{kim2010} to obtain Butman's lower bound, and to validate the Conjecture.} \cite[Theorem~4.1 and Theorem~6.1]{kim2010},  reproduce Butman's lower bound on feedback capacity.}

\subsection{Main Results on the Linear Code of \cite{butman1976}}
\label{sect:mr}

We prove the following conclusion. 
\begin{description}


\item[(C)] Butman's \cite[Abstract]{butman1976} calculation of the rate   $ \lim_{n \longrightarrow \infty} C_{n}^L(\kappa,v_0)$ {for the AR$(c;v_0), c \in [-\infty, \infty]$ noise}, known as Butman's  lower  bound and conjecture on feedback capacity, based on coding scheme (\ref{butman_cs_1_new}),    corresponds to an unbounded   sequence $g_1, g_2, \ldots$,  in the sense  that,  $\lim_{n \longrightarrow \infty}|g_n| =+\infty$. 
\end{description}

Theorem~\ref{thm:ip} presents some of  the consequences of  (C).  To prove (C) we will make use of  Theorem~\ref{thm:bg_0}  (below), known as the ratio test theorem \cite[Theorem 3.34]{Book_analysis:1976}.

\begin{theorem}[The Ratio Test] \ \\
\label{thm:bg_0}
Consider any  sequence of real numbers $\{a_n: n=1,2, \ldots \}$.\\
(a) Suppose that $\lim_{n\longrightarrow \infty} \left| {a_{n+1}}/{a_n} \right| = L$. If $L<1$, then the series $\sum_{n=0}^\infty {a_n}$ converges absolutely, if $L>1$ the series diverges, and if $L=1$ this test gives no information.\\
(b) If $\lim_{n\longrightarrow \infty} \left| {a_{n+1}}/{a_n} \right|  >1$, then $\lim_{n\longrightarrow\infty}|a_n| =+\infty$;  if  $\lim_{n\longrightarrow \infty} \left| {a_{n+1}}/{a_n} \right| <1$, then $\lim_{n\longrightarrow\infty}|a_n| =0$.  
\end{theorem}

Further, {as mentioned earlier},  to  provide   additional insight, we present  an alternative derivation of Butman's lower bound $C_{n}^L(\kappa,v_0)$ for any $n=1,2, \ldots$,  as stated in  Theorem~\ref{thm_bg_1} (below), which is also valid for the  AR$(c_t;v_0), c_t \in (-\infty, \infty)$ noise.

\begin{theorem}[Characterizations of $C_n^L(\kappa, v_0)$ and $C_n^L(\kappa, {\bf P}_{V_1})$] \ \\
\label{thm_bg_1}
Consider the AGN channel defined by
 (\ref{b_ar_1})-(\ref{b_ar_5_a}), {i.e., with AR$(c_t;v_0)$ or AR$(c_t)$, for any  $c_t \in (-\infty, \infty), t=0,1,\ldots$.}   \\
(a) Total Average Power Constraint.  The maximization over all linear coding schemes,  (\ref{butman_cs_1_new}) (without initial state, i.e., for AR$(c_t)$ noise) of mutual information $I(\Theta; Y^n)$,  subject to total average power, $\frac{1}{n}{\bf E}\big\{\sum_{t=1}^n \big(X_t\big)^2\big\}\leq \kappa$,    is given  by 
\begin{align}
&  C_n^{L}(\kappa, {\bf P}_{V_1}) \tri  \sup_{g_t, \: t=1, \ldots, n: \: \: \frac{1}{n}\big\{g_1^2 K_\Theta+ \sum_{t=2}^n g_t^2 \Sigma_{t-1}  \big\} \leq \kappa} \frac{1}{2}\Big\{ \log \Big( \frac{g_1^2 K_\Theta+K_{V_1}}{K_{V_1}}\Big) \nonumber \\
&\hst \hst\hst\hst +\sum_{t=2}^n \log \Big( \frac{\Big(1-c_t \frac{g_{t-1}}{g_t}\Big)^2g_t^2 \Sigma_{t-1} +K_{W_t}}{K_{W_t}}  \Big)\Big\} , \\
 &\mbox{subject to $\Sigma_t, t=1, \ldots, n$ that satisfies the recursion  and controlled by $g_1, \ldots, g_n$,} \nonumber \\
 &\Sigma_t= \frac{K_{W_t}\Sigma_{t-1}}{\Big(g_t-c_t g_{t-1}\Big)^2 \Sigma_{t-1} +K_{W_t}}, \hst t=2, \ldots, n, \label{inn_b_5_in}   \\
&\Sigma_1= \frac{K_{\Theta}K_{V_1}}{g_1^2 K_{\Theta}+K_{V_1}}. \label{inn_b_6_in} 
\end{align}
(b) Pointwise  Average Power Constraint.  The maximization over all linear coding schemes, (\ref{butman_cs_1_new})    (for AR$(c_t)$ noise)   of mutual information $I(\Theta; Y^n)$,  subject to pointwise average power, ${\bf E}\big\{\big(X_t\big)^2\big\}\leq \kappa_t, t=1, \ldots, n$,    is given  by 
\begin{align}
&  C_n^{L}(\kappa_1, \ldots, \kappa_n, {\bf P}_{V_1}) \tri  \sup_{g_t, \: t=1, \ldots, n: \: \: g_1^2 K_\Theta\leq \kappa_1, \:g_t^2 \Sigma_{t-1}  \leq \kappa_t, \:t=2, \ldots, n} \frac{1}{2}\Big\{ \log \Big( \frac{g_1^2 K_\Theta+K_{V_1}}{K_{V_1}}\Big) \nonumber \\
&\hst \hst\hst\hst \hst\hst+\sum_{t=2}^n \log \Big( \frac{\Big(1-c_t \frac{g_{t-1}}{g_t}\Big)^2g_t^2 \Sigma_{t-1} +K_{W_t}}{K_{W_t}}  \Big)\Big\} , \\
 &\mbox{subject to $\Sigma_t, t=1, \ldots, n$ that  satisfies recursion (\ref{inn_b_5_in}), (\ref{inn_b_6_in}) and controlled by $g_1, \ldots, g_n$.} 
\end{align}
(c) The statements of parts (a) and (b) hold, when the noise is replaced by AR$(c_t;v_0)$, with  $V_0=v_0$  the initial state known to the encoder, and with $K_{V_1}$ replaced by $K_{W_1}$, and 
$C_n^{L}(\kappa, {\bf P}_{V_1})=C_n^{L}(\kappa,v_0), C_n^{L}(\kappa_1, \ldots, \kappa_n, {\bf P}_{V_1})=C_n^{L}(\kappa_1, \ldots, \kappa_n, v_0)$.
\end{theorem} 

\begin{proof}
The proof is given  in Section~\ref{lite}. 
\end{proof}

In Section~\ref{sect:linear_code},   we show that Butman's lower  bound of the AR unit memory noise, corresponds to  parameters $K_{W_t}=K_W, c_t=c \in [-1,1], t=1,2 \ldots$, Gaussian message $\Theta \in G(0,1)$, i.e, $K_{\Theta}=1$,  and it is  obtained from Theorem~\ref{thm_bg_1}.(c),  i.e., $C_n^{L}(\kappa_1, \ldots, \kappa_n, v_0)$,   by invoking  Butman's strategy \cite[4 lines below eqn(19)]{butman1976}, 
$sgn(g_n) =-sgn(c g_{n-1}),  n=2,3$.

\subsection{Observations on the  Lower Bound  on Feedback Capacity of \cite{butman1976,wolfowitz1975}  }
\label{sect:linear_code}
Now,  we turn our attention to  Butman's assumptions, formulation, and the  main steps of  the derivation of the lower bound   in  \cite{butman1976}, to verify our observation and claims. It is noted that  Butman's  lower bound is also derived by Wolfowitz \cite{wolfowitz1975}, by transmitting one of $e^{n R}$ messages and maximizing the operational rate $R$ (i.e., without using information theoretic measures of mutual information etc); rather, by applying  the operational  definition of achievable rates.

Below, we make an observation which  is easy to verify.

\begin{observation} By evaluating $C_n^{L}(\kappa_1, \ldots, \kappa_n, v_0)$, described in Theorem~\ref{thm_bg_1}.(c),  for the  parameters $K_{W_t}=K_W, c_t=c \in [-1,1], t=1,2 \ldots$, Gaussian message $\Theta \in G(0,1)$, i.e, $K_{\Theta}=1$,  at Butman's strategy  \cite[4 lines below eqn(19)]{butman1976}, 
\begin{align}
sgn(g_n) =-sgn(c g_{n-1}),  n=2,3, \ \  \mbox{ implies} \ \ \sum_{j=2}^{n}\big(g_j-c g_{j-1}\big)^2= \sum_{j=2}^{n} (g_j)^2 \Big(1+|c| | \frac{g_{j-1}}{g_j}|\Big)^2,
\end{align}
then  by simple algebra we obtain  Butman's and Wolfowitz's lower bound  (this is also  verified in the sequel). 
\end{observation}

{\it Fact 1.}  Butman \cite{butman1976}  and also Wolfowitz \cite{wolfowitz1975},   considered the AGN channel driven by the  AR$(c), c\in [-1,1]$ noise, i.e., $W_t \in G(0,K_W) \forall t$,  and  applied  the linear {\it time-varying coding strategy}\footnote{A coding strategy is  called  {\it time-varying coding strategy} if the  strategy $g_n$ that is used to generate $X_n$ is not fixed for all $n=1, 2, \ldots$.} of a Gaussian message $\Theta \in G(0, 1)$, i.e., zero mean and unit variance $K_\Theta=1$,  to obtain\footnote{The derivation of $\Sigma_t$ is given in Corollary~ \ref{cor_batman_gen_a}, and Theorem~\ref{thm_batman_gen} if no initial state $V_0=v_0$ is assumed.}  
\begin{align}
&X_n= g_n\Big(\Theta - {\bf E}\Big\{\Theta\Big|Y^{n-1}, V_0=v_0\Big\}\Big), \hst n=2,3,\ldots, \hst X_1=  g_1\Theta,  \label{butman_cs_1}
 \\
&\frac{1}{n} {\bf E}\Big\{\sum_{t=1}^n\big(X_t\big)^2\Big|V_0=v_0\Big\} = \frac{1}{n} \sum_{t=1}^n \kappa_t \leq \kappa,\hst  \kappa_t \tri g_t^2 \Sigma_{t-1}, \hso t=2, \ldots, n, \hso \kappa_1 \tri g_1^2,\label{butman_cs_2}\\
&\Sigma_t \tri  {\bf E}\Big\{\Big(\Theta - {\bf E}\Big\{\Theta\Big|Y^{t}, V_0=v_0\Big\}\Big)^2\Big\}  = \frac{K_W }{g_1^2  +\sum_{j=2}^{t}\big(g_j-c g_{j-1}\big)^2+K_W  },  \hso t=2,\ldots, n, \label{butman_cs_2_a}\\
& \Sigma_1=\frac{K_W}{g_1^2+K_W}\label{butman_cs_2_b}
\end{align}
{\it Fact 2.}  \cite[4 lines below eqn(19)]{butman1976},  restricted   the  strategy $g_n$ to the one that  increases the signal-to-noise ratio, $\frac{g_1^2  +\sum_{j=2}^{t}\big(g_j-c g_{j-1}\big)^2}{K_W}$,  as follows. 
\begin{align}
  sgn(g_n) =-sgn(c g_{n-1}), \hst n=2,3, \ldots \hso \Longrightarrow \hso \sum_{j=2}^{n}\big(g_j-c g_{j-1}\big)^2= \sum_{j=2}^{n} (g_j)^2 \Big(1+|c|| \frac{g_{j-1}}{g_j}|\Big)^2.  \label{butman_cs__a2}
\end{align} 
{\it Fact 3.} \cite[eqn(26)--eqn(28)]{butman1976}  evaluated  $I(\Theta; Y^n|v_0)$, by applying  the linear coding strategy  (\ref{butman_cs__a2})    as follows. 
\begin{align}
    I(\Theta; Y^n|v_0) =& \frac{1}{2}\log \Big\{1+ \frac{\kappa_1}{K_W} \Big\}+   \frac{1}{2}\sum_{t=2}^n \log \Big\{1+ \Big(g_t-c g_{t-1}\Big)^2 \frac{\Sigma_{t-1}}{K_W} \Big\}  \label{ihara_8_aaa}  \\
=& \frac{1}{2}\log \Big\{1+ \frac{\kappa_1}{K_W} \Big\}+   \frac{1}{2}\sum_{t=2}^n \log \Big\{1+ \Big(1+ \frac{|c|}{\chi_t}\Big)^2 \frac{\kappa_t}{K_W} \Big\}    \hst \mbox{by (\ref{butman_cs__a2})} \label{ihara_8_aa}
\end{align}
where $\chi_n$ is related to $(g_n, g_{n-1})$ and satisfies the recursion \cite[eqn(23)]{butman1976},
\begin{align}
& \chi_n \tri \Big|\frac{g_n}{g_{n-1}}\Big|, \hst n=3, 4, \ldots, \hst \chi_2= 1+ \frac{\kappa_1}{K_W}, \label{ihara_8}\\
&\chi_n^2=\Big\{1+\Big(1+ \frac{|c|}{\chi_{n-1}}\Big)^2 \frac{\kappa_{n-1}}{K_W}\Big\}\frac{\kappa_n}{\kappa_{n-1}}, \hso n=3,4, \ldots  \label{ihara_9}
\end{align}
{\it Fact 4.} Butman's  resulting optimization problem, from Fact 3, reduces  to  
\begin{align}
{\bf (B1):} \hst 
 &\frac{1}{n}C_n^{L,B1}(\kappa,v_0)\tri  \sup_{g_1, \ldots, g_n}\Big\{
 \frac{1}{2n}\log \Big\{1+ \frac{\kappa_1}{K_W} \Big\}+   \frac{1}{2n}\sum_{t=2}^n \log \Big\{1+ \Big(1+ \frac{|c|}{\chi_t}\Big)^2 \frac{\kappa_t}{K_W} \Big\} \Big\}
  \label{ihara_8_a}\\
  &\mbox{ subject to} \hso \sum_{t=1}^n \kappa_t \leq \kappa, \hso  \kappa_t \tri g_t^2 \Sigma_{t-1}, \hso t=2, \ldots, n, \hso \kappa_1= g_1^2,\\
  &\mbox{$\chi_t$ satisfies recursion (\ref{ihara_8}), (\ref{ihara_9}),} \label{butman_cs_2_bbb}\\
  &\Sigma_t   = \frac{K_W }{g_1^2  +\sum_{j=2}^{t}(g_j)^2\big(1+\frac{|c|}{\chi_j}\big)^2+K_W  },  \hso t=2,\ldots, n  \hst \mbox{by (\ref{butman_cs__a2})},   \label{butman_cs_2_aa}\\
& \Sigma_1=\frac{K_W}{g_1^2+K_W}.\label{butman_cs_2_bb}
\end{align}
{Notice that in {\bf (B1)} one is asked to optimize over $\{g_1, g_2, \ldots, g_n\}$ subject to constraints, or over $\{\kappa_1, \kappa_2, \ldots, \kappa_n\}$}. 

{\it Fact 5.}  \cite{butman1976} did not provide a solution to $C_n^{L,B1}(\kappa,v_0)$;  instead the statement of Butman's  conjecture \cite[Abstract]{butman1976} on feedback capacity is based on a variation, based on Assumptions~\ref{ass:a1} (below).

\begin{assumptions}  \cite{butman1976}  Average point-wise power constraints\footnote{Wolfowitz's  \cite{wolfowitz1975} derivation is also based on (\ref{bt_ass_1}).}  \\
\label{ass:a1}
  Butman's lower bound on feedback capacity and conjecture \cite[Abstract]{butman1976},  is based on the  single-letter   average  power at the transmitter (easily verified from \cite{butman1976}), given by 
\begin{align}
{\bf E}\Big\{\big(X_t\big)^2\Big|V_0=v_0\Big\} = \kappa, \hso  t=1,2, \ldots, n \hso \Longrightarrow \hso \kappa_t \tri g_t^2 \Sigma_{t-1}=\kappa, \hso t=2, \ldots, n, \hso \kappa_1 \tri g_1^2 =\kappa \label{bt_ass_1}
\end{align}
where $\Sigma_t, t=1, \ldots, n$ satisfies (\ref{butman_cs_2_aa}) and (\ref{butman_cs_2_bb}).
\end{assumptions}

{In view of Assumptions~\ref{ass:a1}, Butman's lower bound  \cite[see transition from eqn(23) to eqn(23a) or   eqn(28) and paragraph above it]{butman1976}, is given as follows\footnote{There is no optimization over $g_1, g_2, \ldots$ because Assumptions~\ref{ass:a1} imply $\chi_1,\chi_2, \ldots$ does not depend on $g_1, g_2, \ldots$.}. } 
\begin{align}
{\bf (B2):} \hst  &\frac{1}{n}C_n^{L,B2}(\kappa,v_0)=  
 \frac{1}{2n}\log \Big\{1+ \frac{\kappa}{K_W} \Big\}+   \frac{1}{2n}\sum_{t=2}^n \log \Big\{1+ \Big(1+ \frac{|c|}{\chi_t}\Big)^2 \frac{\kappa}{K_W} \Big\}
  \label{ihara_8_aa}\\
  &\mbox{$\chi_t$ satisfies recursion $\chi_t^2=1+\Big(1+ \frac{|c|}{\chi_{n-1}}\Big)^2 \frac{\kappa}{K_W}, \hso t=3,4, \ldots, \chi_2= 1+ \frac{\kappa}{K_W}$,}\label{ihara_8_aaa}\\
  &  g_t^2 \Sigma_{t-1}=\kappa, \hso t=2, \ldots, n, \hso g_1^2=\kappa, \hso 
  \mbox{$\Sigma_t, t=1, \ldots, n$ satisfies (\ref{butman_cs_2_aa}), (\ref{butman_cs_2_bb}).} 
\end{align}
{Unlike {\bf (B1)}, in {\bf (B2)}  there is no  optimization  over $\{g_1, g_2, \ldots, g_n\}$  or $\{\kappa_1, \kappa_2, \ldots, \kappa_n\}$.}

{\it Fact 6.} Butman's lower bound \cite[Abstract, with $m=1$ or eqn(4), eqn(5), eqn(11)]{butman1976}
Butman \cite[eqn(23a) and paragraph above eqn(28)]{butman1976}, which is based on  Assumptions~\ref{ass:a1} and strategy  (\ref{butman_cs__a2}), is stated as follows.   
\begin{align}
{\bf (B):} \hst  C^{L,B}(\kappa)  \tri \lim_{n \longrightarrow \infty} \frac{1}{n}   C_n^{L,B2}(\kappa,v_0) 
 =\frac{1}{2}\log \chi^2,  \label{butman_cs_2}
 \end{align}
where 
\begin{align}
&\chi = \lim_{n \longrightarrow \infty} \chi_n= \lim_{n \longrightarrow \infty} \Big| \frac{g_n}{g_{n-1}}\Big|,  \label{butman_cs_4}\\
&\mbox{$\chi$ is the positive root of} \hso 
\chi^4-\chi^2 - \frac{\kappa}{K_W} \Big(\chi+|c|\Big)^2=0, \hst |c|\leq 1, \hso K_W>0, \hso \kappa \geq 0, \label{butman_cs_3}
\end{align}
Butman \cite{butman1976}   conjectured that $C^{L,B}(\kappa)$ is the feedback capacity. 

%
%
%

\begin{remark}
 A comparison of $C^{L,B}(\kappa)$  to an upper bound derived by  Tienan and Schalkwijk \cite[Section~I]{tienan-schalkwijk1974}   is discussed in \cite[eqn(1) and eqn(2)]{butman1976}. 
Ozarow \cite{ozarow1990} and Dembo \cite{dembo1989} re-visited Butman's lower bound and derived upper bounds on feedback rates. 
\end{remark}

To complete our observations on Butman's  problems we note two more facts.

{\it Fact 7.} If we do not impose Butman's restriction    that strategies  $g_n$ satisfy $ sgn(g_n) =-sgn(c g_{n-1}), n=2,3, \ldots$, i.e., (\ref{butman_cs__a2}) is not assumed, then the optimization problem {\bf (B1)} and statement {\bf (B2)}, are replaced by {\bf (P1)} and {\bf (P2)}, respectively, given   below.
\begin{align}
{\bf (P1):} \hst 
 &\frac{1}{n}C_n^{L,P1}(\kappa,v_0)\tri  \sup_{g_1, \ldots, g_n}\Big\{\frac{1}{2n}\log \Big\{1+ \frac{\kappa_1}{K_W} \Big\}+   \frac{1}{2n}\sum_{t=2}^n \log \Big\{1+ \Big(g_t-c g_{t-1}\Big)^2 \frac{\Sigma_{t-1}}{K_W} \Big\} \Big\}
  \label{ihara_8_ap}\\
  &\mbox{ subject to} \hso \frac{1}{n} \sum_{t=1}^n \kappa_t \leq \kappa, \hso  \kappa_t \tri g_t^2 \Sigma_{t-1}, \hso t=2, \ldots, n, \hso \kappa_1= g_1^2,\\
  &\Sigma_t   = \frac{K_W }{g_1^2  +\sum_{j=2}^{t}\big(g_j-c g_{j-1}\big)^2+K_W  },  \hso t=2,\ldots, n,   \label{butman_cs_2_aap}\\
& \Sigma_1=\frac{K_W}{g_1^2+K_W},\label{butman_cs_2_bbp}
\end{align}
\begin{align}
{\bf (P2):} \hst 
 &\frac{1}{n}C_n^{L,P2}(\kappa,v_0)\tri  \sup_{g_1, \ldots, g_n}\Big\{\frac{1}{2n}\log \Big\{1+ \frac{\kappa_1}{K_W} \Big\}+   \frac{1}{2n}\sum_{t=2}^n \log \Big\{1+ \Big(g_t-c g_{t-1}\Big)^2 \frac{\Sigma_{t-1}}{K_W} \Big\} \Big\}
  \label{ihara_8_app}\\
  &\mbox{ subject to} \hso  g_t^2 \Sigma_{t-1}=\kappa, \hso t=2, \ldots, n, \hso g_1^2=\kappa,\\
  &\Sigma_t,  t=1,\ldots, n \hso \mbox{satisfies   (\ref{butman_cs_2_aap}), (\ref{butman_cs_2_bbp}).}  \label{butman_cs_2_aapp}
\end{align}
{Notice {\bf (P1),  (P2)} are optimization problems, where   $\{g_1, g_2, \ldots, g_n\}$ control $\{\Sigma_1, \Sigma_2, \ldots, \Sigma_n\}$.}\\
{In Section~\ref{sect:sim} we compare the numerical solutions of  {\bf (B2)} and {\bf (B)} to {\bf  (P2)}. }

The next theorem, is an application of the ratio test theorem to Butman's lower bound $C^{L,B}(\kappa)$ given by  (\ref{butman_cs_2})-(\ref{butman_cs_4}).

\begin{theorem} On Butman's lower bound\\
\label{thm:ip}
Consider the lower bound  \cite{butman1976},   $C^{L,B}(\kappa)=\frac{1}{2}\log \chi^2$ given by  (\ref{butman_cs_2})-(\ref{butman_cs_3}). \\
 If $\lim_{n\longrightarrow \infty} \left| {g_{n}}/{g_{n-1}} \right| =\chi   >1$, then $\lim_{n\longrightarrow\infty}|g_n| =+\infty$. \\ 
 If $\lim_{n\longrightarrow \infty} \left| {g_{n}}/{g_{n-1}} \right| =\chi   < 1$, then $\lim_{n\longrightarrow\infty}|g_n| =0$, and  $C^{L,B}(\kappa)=\frac{1}{2}\log \max\{1, \chi^2\}=0 \; \forall \kappa \in [0,\infty)$.
\end{theorem} 
\begin{proof}
Suppose Butman's  sequence $\{g_1, g_2, \ldots, g_n\}$,   corresponding to   (\ref{butman_cs_4}), (\ref{butman_cs_3}),  i.e.,  computed using $\chi = \lim_{n \longrightarrow \infty} \chi_n= \lim_{n \longrightarrow \infty} \Big| \frac{g_n}{g_{n-1}}\Big|$,  such that  $\chi$ is 
 is the real positive root of the   quadric equation (\ref{butman_cs_3}),   with   $|\chi|>1$. Then by Theorem~\ref{thm:bg_0}.(b), sequence $\{g_1, g_2, \ldots, g_n\}$, is such that,  $|g_n| \longrightarrow \infty$,  as  $n\longrightarrow \infty$.  On the other hand, Theorem~\ref{thm:bg_0}.(b),   if $|\chi |<1$ then $|g_n| \longrightarrow 0$,  as  $n\longrightarrow \infty$,  and since $C^{L,B}(\kappa)$ is computed using mutual information which takes values in $[0,\infty]$ then the claim holds.
\end{proof} 

This implies, Butman 's  lower bound is achieved by an  unbounded sequence of  coding gains $|g_1|, |g_2|, \ldots$. Therefore, such coding schemes are not practical, and moreover they are  {\it fragile},  at least for transmission intervals $\{1,2, \ldots, n\}$ of  moderate duration $n$. By fragile, we mean,  any mismatch of the estimation error $\Theta - {\bf E}\Big\{\Theta\Big|Y^{n-1}, V_0=v_0\Big\}$ due to any additional external noise in the Gaussian channel not accounted for, at any given time of transmission, will be amplified by the scaling $g_n$ (see abstract).

Further to the above, it should be apparent that optimization problem $(\mathbf{B1})$, which uses a total average power constraint,  i.e.,  (\ref{ihara_8_a})-(\ref{butman_cs_2_bbb}),  is also  achieved by an unbounded sequence $|g_1|, |g_2|, \ldots$.


\begin{remark} 
From the above  follows that  convergence of the sequence $\chi_n, n=1, 2, \ldots$ does not imply convergence of the sequence $|g_n|, n=1,2, \ldots$. In fact, if $|g_n|, n=1,2, \ldots$ converges to $|g|$ then necessarily $\chi_n, n=1,2, \ldots $ convergences to $\chi =1$, and the value of the lower bound,  is $\frac{1}{2} \log \chi^2=0$. \\
\end{remark}

\begin{figure}[ht]
    \centering
    \includegraphics[width=0.8\columnwidth]{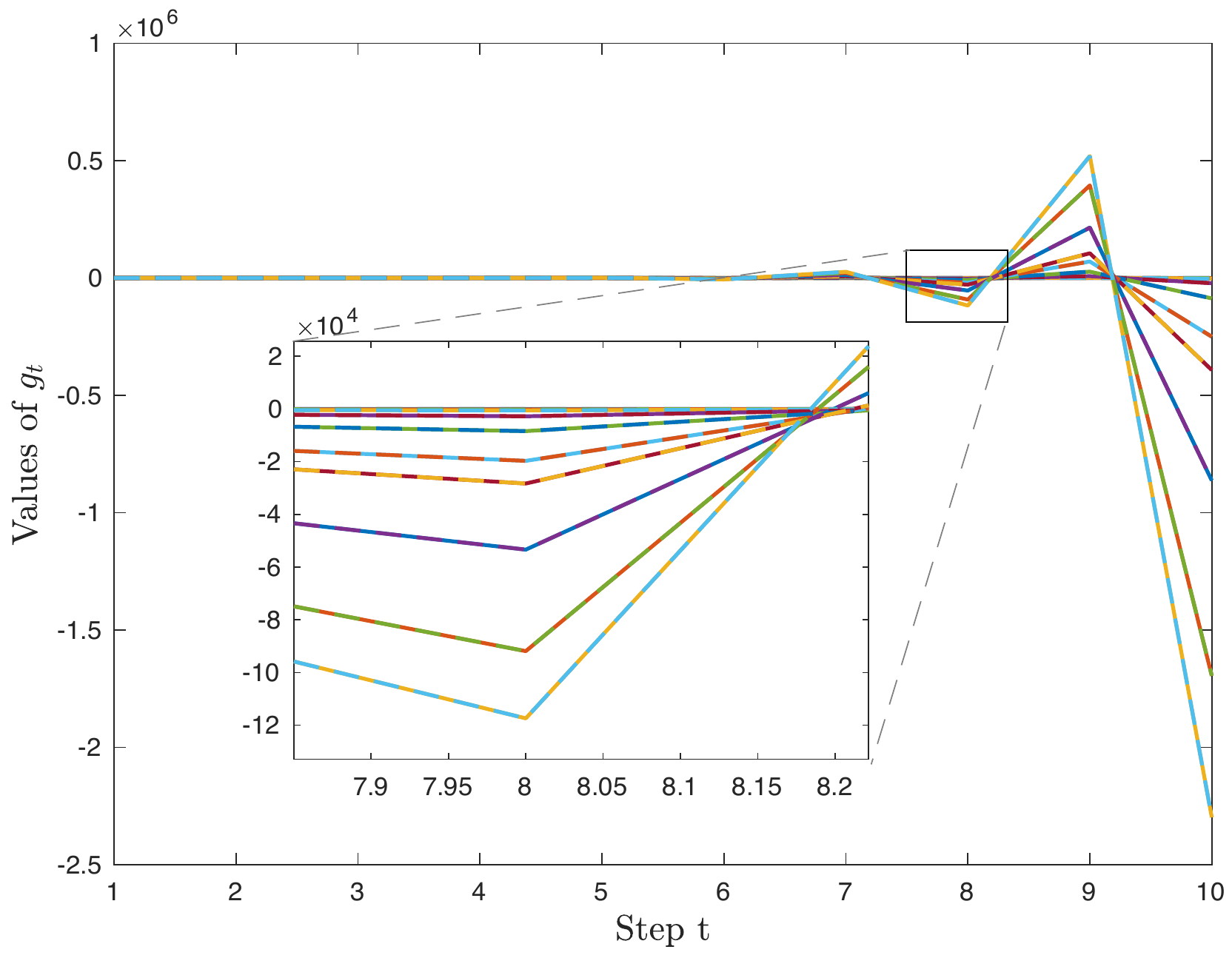}
    \caption{Numerical example with which it is shown that Butman's strategy  $g_t,t=1,2,\ldots$ , shown in solid line, as well as the numerical evaluation  of {\bf (P2)}, shown in dashed line,  for $n=10$ have values of $g_t, t=1, \ldots, n$ that are similar and grow unbounded. The small differences, however, lead to the gap in the solution, as shown in Fig.~\ref{fig:comparison2}, further supporting our claim about the fragility of Butman's strategy $g_1,g_2, \ldots$.  For the simulations, we used $K_{\Theta} = K_W = K_V =1$, $c =0.5$. The values of the sequence $|g_1|, |g_2|, \ldots, |g_n|$ grows unbounded, as  predicted by Theorem~\ref{thm:ip}. }
    \label{fig:comparison1}
\end{figure}

\subsection{Simulations of Butman's Coding Scheme}
\label{sect:sim}
{Figure~\ref{fig:comparison2} is a  comparison between the numerical evaluation of {\bf (B2)}, {\bf (B)}  and {\bf (P2)}. Notice that {\bf (B2)} and  {\bf (B)} are derived by using Butman's strategy (\ref{butman_cs__a2}).  For {\bf (B2)}, once $\chi_t, t=1,2, \ldots, n$ is computed then   $g_t^2 \Sigma_{t-1}=\kappa, t=2, \ldots, n,  g_1^2=\kappa,$ and $g_t, t=1, \ldots, n$ is found from $\Sigma_t, t=1, \ldots, n$ that satisfies (\ref{butman_cs_2_aa}), (\ref{butman_cs_2_bb}) }

While the values of the sequence $g_1, g_2, \ldots, g_n$ are similar, and their absolute values  grow unbounded, as it is shown in Fig.~\ref{fig:comparison1}, the numerical evaluation of optimization problem {\bf (P2)} for $n=10$ performs better than    Butman's lower bound {\bf (B2)}  (which is based on Butman's strategy (\ref{butman_cs__a2})), even for $n=20$.  Therefore, for finite $n$, simulations show that  Butman's strategy (\ref{butman_cs__a2})  leading to {\bf (B2)}, is not optimal. Note that for values of $n>10$,     the numerical optimization problem {\bf (P2)}  is difficult to complete, because the sequence $|g_n|$ grows unbounded,  and the numerical optimization does not converge for the maximum number of iterations considered.  Similarly, it is difficult to determine $g_t, t=1,2, \ldots, n$ of  {\bf (B2)} for large $n$.  Nevertheless, the main point that Butman's strategy (\ref{butman_cs__a2})  is not optimal for finite $n$,  can be inferred from the simulations of {\bf (P2)} for $n=10$.

\begin{figure}[ht]
    \centering
    \includegraphics[width=0.8\columnwidth]{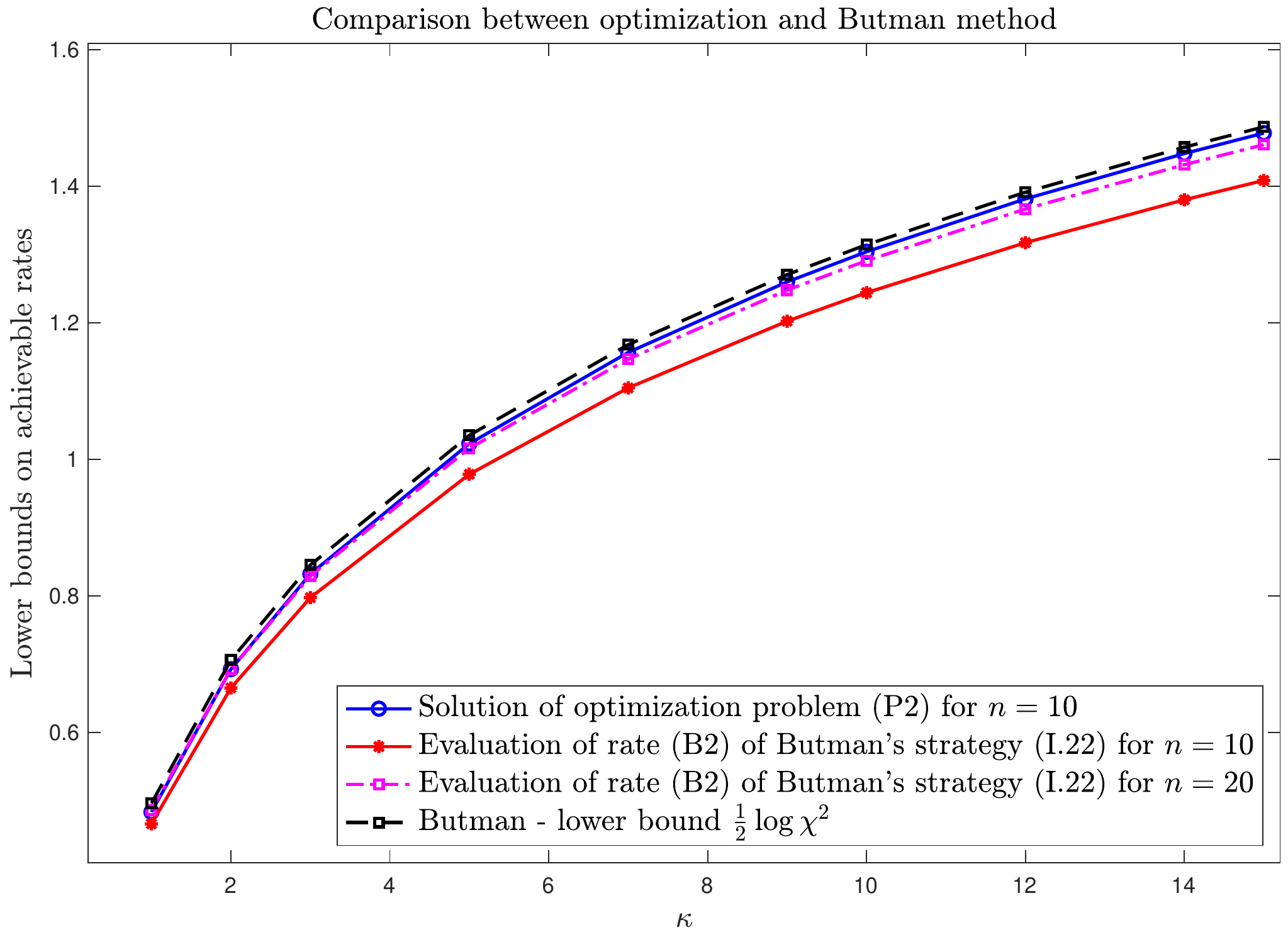}
    \caption{Numerical comparisons of rates based on {\bf (B2)}, {\bf (P2)} and asymptotic limit, called Batman's lower bound  based on {\bf (B)}. Problem {\bf (P2)} is solved using standard MATLAB optimization functions (for $n=10$). For the simulations, we used $K_{\Theta} = K_W = K_V =1$, $c =0.5$. The asymptotic for Butman's strategy seems to be optimal as $n\to \infty$, but it is not optimal for a finite $n$.}
    \label{fig:comparison2}
\end{figure}

On the other hand, the calculation of  asymptotic limit based on {\bf (B)},  gives a value which is higher than {\bf (B2)} calculated for $n=10$ and $n=20$.  This is expected, because the rate based on  {\bf (B2)}, i.e., $\frac{1}{n}C_n^{L,B2}(\kappa, v_0) $,   is nondecreasing with $n$.  On other hand, the   rate based on  {\bf (P2)}, i.e., $\frac{1}{n}C_n^{L,P2}(\kappa, v_0) $,   is also  nondecreasing with $n$,  but it  was not possible to compute it for large values of $n$, because i) $|g_1|, |g_2|, \ldots $ grows unbounded,  and  it is difficult to compute it numerically, even for moderate values of  $n$ beyond $10$, and ii) an  analytic expression of  {\bf (P2)} is not available.

\begin{observation} On Butman's strategy    (\ref{butman_cs__a2})  \\
Simulations show that Butman's strategy $ sgn(g_n) =-sgn(c g_{n-1})$, $n=2,3, \ldots$, i.e.,  which is used to obtain the recursion $\chi_n, n=1,2, \ldots$ is not optimal for finite $n$, since the numerical optimization problem {\bf (P2)} produces another sequence $g_n$ with higher value  of rate  (see Fig.~\ref{fig:comparison2}).  These simulations highlight the severe limitations of Butman's scheme, since it is sub-optimal for small number of transmissions $n$, and highly impractical for large number of transmissions $n$, because $|g_1|, |g_2|, \ldots$ grows  unbounded. 
\end{observation}

 \section{Independent Derivation of  Butman's Lower Bound and  Additional Discussion}
\label{lite}
In this section we provide the derivation of Theorem~\ref{thm_bg_1}, and then we show how to recover, as degenerate cases,  Butman's lower bounds.

We consider the AGN channel defined by (\ref{b_ar_1})- (\ref{b_ar_5_a}), without an initial state, and we derive the characterization of $C_n^L(\kappa, {\bf P}_{V_1})$, with coding scheme  defined by (\ref{ir})  and (\ref{butman_cs_1_new}), respectively,  (without an initial state).  The analysis below, shows that, Butman's equations can be obtained for any time-varying noise with $c_t \in (-\infty, \infty)$.

\begin{theorem} Preliminary characterization of $C_n^L(\kappa, {\bf P}_{V_1})$ \\ 
\label{thm_batman_gen}
Consider   the AGN, driven by a time-varying  AR$(c_t), t\in (-\infty,\infty)$ noise,  without initial state.\\
 Define  the  conditional mean and error covariance by
\begin{align}
\widehat{\Theta}_{t} \tri & {\bf E}\Big\{\Theta\Big|Y^{t}\Big\}, \hst  \Sigma_t\tri {\bf E}\Big\{ \Big(\Theta - \widehat{\Theta}_t\Big)^2\Big|Y^t\Big\}, \hso t=1, \ldots, n.
\end{align}
Consider the  linear coding scheme
\begin{align}
X_t=& g_t\Big(\Theta - \widehat{\Theta}_{t-1}\Big), \hso t=2, \ldots, n, \hso X_1=g_1 \Theta, \label{inn_b_a}  \\
Y_t =& g_t\Big(\Theta - \widehat{\Theta}_{t-1}\Big) +V_t,   \hso t=2, \ldots, n,  \\
=& g_t\Big(\Theta - \widehat{\Theta}_{t-1}\Big) -c_tg_{t-1}\Big(\Theta - \widehat{\Theta}_{t-2}\Big) +c Y_{t-1}+ W_t  \hst \mbox{by $V_{t-1}=Y_{t-1}-X_{t-1}$} \label{inn_b_b}\\
 Y_1=& g_1\Theta  +V_1, \\
{\bf E}&\big\{\sum_{t=1}^n \big(X_t\big)^2\big\}= \frac{1}{n} \sum_{t=1}^n \kappa_t \leq \kappa,  \hst  \kappa_t =g_t^2 \Sigma_{t-1}, \hso t=2, \ldots, n,\hst \kappa_1=g_1^2 K_{\Theta}.  \label{inn_b_c}
\end{align}
Then the following hold.\\
(a) The innovations process of $Y^n$ denoted by $I^n$ is   an  orthogonal Gaussian  process, $I_t \in G(0, K_{I_t}), t=1, \ldots$,  given by
 \begin{align}
 I_t \tri & Y_t -{\bf E}\Big\{Y_t\Big|Y^{t-1}\Big\}, \hst t=2, \ldots, n \label{inn_b_1}    \\
 =&\Big(g_t -c_t g_{t-1}\Big)\Big(\Theta - \widehat{\Theta}_{t-1}\Big)+W_t,\\
 I_1=& g_1 \Theta +V_1. \label{inn_b_2}
 \end{align}
 where 
 \begin{align}
 K_{I_t} \tri & {\bf E} \Big\{ \big(I_t\Big)^2\Big\}
=  \Big(g_t-c_t g_{t-1}\Big)^2 {\bf E}\Big\{ \Big( \Theta - \widehat{\Theta}_{t-1}\big\}\Big)^2\Big\}  + K_{W_t}, \hst t=2,3, \ldots,  \label{inn_b_3}\\
=&  \Big(g_t-c_t g_{t-1}\Big)^2\Sigma_{t-1}+K_{W_t}\\
K_{I_1} \tri & {\bf E}\Big\{\big(I_1\Big)^2\Big\}
=g_1^2K_{\Theta} + K_{V_1} \label{inn_b_4}
\end{align}  
and where the mean-square error     $\Sigma_t$ and estimate $\widehat{\Theta}_t$ satisfy the recursions 
\begin{align}
&\Sigma_t= \frac{K_{W_t}\Sigma_{t-1}}{\Big(g_t-c_t g_{t-1}\Big)^2 \Sigma_{t-1} +K_{W_t}}, \hst t=2, \ldots, n, \label{inn_b_5}   \\
&\Sigma_1= \frac{K_{\Theta}K_{V_1}}{g_1^2 K_{\Theta}+K_{V_1}}, \label{inn_b_6} \\
&\widehat{\Theta}_t= \widehat{\Theta}_{t-1} +\frac{ \Big(g_t-c_t g_{t-1}\Big)\Sigma_{t-1}}{\Big(g_t-c_t g_{t-1}\Big)^2\Sigma_{t-1}+K_{W_t}}I_t, \hst t=2, \ldots, n, \label{inn_b_7} \\
&\widehat{\Theta}_1=\frac{ g_1 K_{\Theta}}{g_1^2K_{\Theta}+K_{V_1}}I_1. \label{inn_b_8}
\end{align}
 The mutual information between the Gaussian RV $\Theta$ and $Y^n$,  is given by
 \begin{align}
I(\Theta; Y^{n}) =& H(Y^n)-  H(V^n)\label{inn_b_9}\\
=&\sum_{t=1}^n H(I_t)- H(V_1) -\sum_{t=2}^n H(W_t) \label{inn_b_10}\\
=&\frac{1}{2 } \log \Big( \frac{g_1^2 K_{\Theta}  +K_{V_1}}{K_{V_1}} \Big) +\frac{1}{2}\sum_{t=2}^n \log \Big( \frac{\Big(g_t-c_t g_{t-1}\Big)^2 \Sigma_{t-1}+K_{W_t}}{K_{W_t}}  \Big)\label{inn_b_11_a}\\
=&\frac{1}{2 } \log \Big( \frac{ \kappa_1 +K_{V_1}}{K_{V_1}} \Big) +\frac{1}{2}\sum_{t=2}^n \log \Big( \frac{\Big(1-c_t \frac{g_{t-1}}{g_t}\Big)^2 \kappa_t+K_{W_t}}{K_{W_t}}  \Big). \label{inn_b_11}
\end{align}
(b)  Suppose  $K_{W_t}=K_W, t=2, \ldots, n$. Then  
\begin{align}
\Sigma_t =& \frac{K_W K_{\Theta} K_{V_1}}{g_1^2 K_W K_{\Theta} +\sum_{j=2}^{t}\big(g_j-c_jg_{j-1}\big)^2K_{\Theta}K_{V_1}+K_W K_{V_1} },  \hso t=2,\ldots, n, \label{inn_b_13}\\
\Sigma_1 = &\frac{ K_{\Theta} K_{V_1}}{g_1^2 K_{\Theta} + K_{V_1} },\\
{\bf E}\Big\{ & \sum_{t=1}^n \big(X_t\big)^2\Big\}= \frac{1}{n} \sum_{t=1}^n \kappa_t \leq \kappa, \\ 
\kappa_t =& \frac{g_t^2 K_W K_{\Theta} K_{V_1}}{g_1^2K_W K_\Theta+\sum_{j=2}^{t-1} \Big(g_j-c_j g_{j-1}\Big)^2 K_\Theta K_{V_1}+ K_W K_{V_1}}, \hso t=2, \ldots, n, \label{inn_b_14}\\
 \kappa_1 =& \frac{g_1^2 K_{\Theta} K_{V_1}}{g_1^2 K_\Theta+  K_{V_1}}\label{inn_b_14a}\\
I(\Theta; Y^{n})    =&\frac{1}{2 } \log \Big( \frac{g_1^2 K_\Theta +K_{V_1}}{K_{V_1}} \Big) +\frac{1}{2}\sum_{t=2}^n \log \Big(  \big(1-c_t \frac{g_{t-1}}{g_t}\big)^2  \frac{\kappa_t}{K_{W}K_\Theta K_{V_1}}  +1  \Big)\label{inn_b_15} \\
=& \frac{1}{2 } \log \Big( \frac{g_1^2 K_\Theta +K_{V_1}}{K_{V_1}} \Big) +\frac{1}{2}\sum_{t=2}^n \log \Big( \chi_{t+1} \frac{\kappa_t}{\kappa_{t+1}} \Big)\label{inn_b_16} 
\end{align}
where $\chi_t\tri \frac{g_t}{g_{t-1}}$, and 
\begin{align}
\chi_{t+1}= \Big\{1 + \Big(1-c_t \frac{1}{\chi_t}\Big)^2 \frac{\kappa_t}{K_{W_t}K_{\Theta}K_{V_1}}\Big\} \frac{\kappa_{t+1}}{\kappa_t}, \hso t=2,\ldots, n, \hso \chi_2= \frac{g_2}{g_1}.\label{inn_b_17} 
\end{align}
\end{theorem}
\begin{proof} (a) By  
\begin{align}
{\bf E}\Big\{Y_t\Big|Y^{t-1}\Big\}=&-c_tg_{t-1} \widehat{\Theta}_{t-1} +c_tg_{t-1} {\bf E}\Big\{\widehat{\Theta}_{t-2}\Big|Y^{t-1}\Big\} + c_tY_{t-1}\\
=& -c_tg_{t-1} \widehat{\Theta}_{t-1} +c_tg_{t-1}\widehat{\Theta}_{t-2} + c_tY_{t-1}
\end{align}
and substituting into $I_t$ we obtain (\ref{inn_b_1}), (\ref{inn_b_2}). Then (\ref{inn_b_3}), (\ref{inn_b_4}) are directly obtained from the independence of $\Theta$ and $W_t, t=2, \ldots, n, V_1$. By mean-square estimation theory, the conditional covariance of $\Theta$ given $Y^t$, denoted by $cov(\Theta, \Theta\Big|Y^{t})$ is,
\begin{align}
cov(\Theta, \Theta\Big|Y^{t})=&cov(\Theta, \Theta\Big|Y^{t-1})-\Big( cov(\Theta,Y_t\Big|Y^{t-1})\Big)^2 \Big( cov(Y_t,Y_t\Big|Y^{t-1})\Big)^{-1} \\
=& cov(\Theta, \Theta\Big|Y^{t-1})-\frac{\Big(g_t-c_tg_{t-1}\Big)^2 \Big(cov(\Theta, \Theta\Big|Y^{t-1})\Big)^2}{ \Big(g_t-c_tg_{t-1}\Big)^2 \Big(cov(\Theta, \Theta\Big|Y^{t-1})\Big)+ K_{W_t}}, \hso t=2, \ldots, n,\\
 cov(\Theta, \Theta\Big|Y_1)=& \frac{K_\Theta K_{V_1}}{g_1^2 K_\Theta +K_{V_1}}. \label{but_proof_1}
\end{align}
From the above follows, $cov(\Theta, \Theta\Big|Y^{t})= {\bf E}\Big\{ \Big(\Theta - \widehat{\Theta}_t\Big)^2\Big\}=\Sigma_t$. Hence, (\ref{but_proof_1}), reduces to
\begin{align}
\Sigma_t=&\Sigma_{t-1} - \frac{\Big(g_t-c_t g_{t-1}\Big)^2 \Sigma_{t-1}^2}{ \Big(g_t-c_t g_{t-1}\Big)^2 \Sigma_{t-1}+ K_{W_t}}, \hso t=2, \ldots, n,    \label{but_proof_2}  \\ \Sigma_1=&\frac{K_\Theta K_{V_1}}{g_1^2 K_\Theta +K_{V_1}}. \label{but_proof_3}
\end{align}
From (\ref{but_proof_2}), (\ref{but_proof_3}) then follows   (\ref{inn_b_5}),  (\ref{inn_b_6}). The conditional  mean of  of $\Theta$ given $Y^t$ is, 
\begin{align}
&{\bf E}\Big\{\Theta\Big|Y^t\Big\}= {\bf E}\Big\{\Theta\Big|Y^{t-1}\Big\}+ \Big( cov(\Theta,Y_t\Big|Y^{t-1})\Big) \Big( cov(Y_t,Y_t\Big|Y^{t-1})\Big)^{-1}\Big(Y_t - {\bf E}\Big\{Y_t\Big|Y^{t-1}\Big\}\Big), \hso t=2, \ldots, n, \label{but_proof_4}    \\
&{\bf E}\Big\{\Theta\Big|Y_t\Big\}={\bf E}\Big\{\Theta\Big\}+ cov(\Theta, Y_1)\Big(cov(Y_1,Y_1)\Big)^{-1}\Big(Y_1 - {\bf E}\Big\{Y_1\Big\}\Big).\label{but_proof_5} 
\end{align}
From  (\ref{but_proof_4}),  (\ref{but_proof_5}) and the conditional variance above, then follows (\ref{inn_b_7}), (\ref{inn_b_8}).  From the above  then follows (\ref{inn_b_9})-(\ref{inn_b_11}). (b) Using the assumptions we deduce   (\ref{inn_b_13})-(\ref{inn_b_15}). To show (\ref{inn_b_16}), (\ref{inn_b_17}) we follow Butman \cite{butman1969}. Let
\begin{align}
SNR_t =&g_1^2 K_W K_\Theta + \sum_{j=2}^t \Big(g_j-c_t g_{j-1}\Big)^2 K_\Theta K_{V_1}, \hso t=1, 2, \ldots, n\\
=& SNR_{t-1} +g_t^2 \Big(1-c\frac{g_{t-1}}{g_t}\Big)^2 K_\Theta K_{V_1}.
\end{align}
Hence, 
\begin{align}
 SNR_t+K_W K_{V_1}=&  SNR_{t-1} +g_t^2 \Big(1-c_t\frac{g_{t-1}}{g_t}\Big)^2 K_\Theta K_{V_1}+K_W K_{V_1}\\
=&\Big( SNR_{t-1}+K_W K_{V_1}\Big)\Big(1 + \frac{g_t^2}{K_W K_{V_1}+ SNR_{t-1}}\big(1-c_t \frac{g_{t-1}}{g_t}\big)^2\Big)\\
=&\Big(SNR_{t-1}+K_W K_{V_1} \Big)\Big(1 + \frac{\kappa_t}{K_WK_\Theta K_{V_1}}\big(1-c_t \frac{g_{t-1}}{g_t}\big)^2\Big).
\end{align}
Letting $\chi_t=\frac{g_t}{g_{t-1}}$, then 
\begin{align}
1 + \frac{\kappa_t}{K_WK_\Theta K_{V_1}} \Big(1-c_t \frac{g_{t-1}}{g_t}\Big)^2= \frac{K_W K_{V_1} +SNR_t}{K_W K_{V_1} +SNR_{t-1}}= \frac{\kappa_t}{\kappa_{t+1}} \chi_{t+1}.
\end{align}
From the above follows (\ref{inn_b_16}),   (\ref{inn_b_17}).
\end{proof}

From Theorem~\ref{thm_batman_gen} follows an analogous preliminary characterization, for the AR$(c_t;v_0), t\in (-\infty,\infty)$ noise, i.e., when $V_0=v_0$ is the initial state known to the encoder, as stated in the next corollary.

\begin{corollary} Preliminary characterization of $C_n^L(\kappa, v_0)$ \\ 
\label{cor_batman_gen_a}
Consider the statement of Theorem~\ref{thm_batman_gen}, with  the AGN channel driven by a time-varying  AR$(c_t;v_o)$ noise.\\
Define  the  conditional mean and error covariance for fixed $V_0=v_0$, by
\begin{align}
\widehat{\Theta}_{t} \tri & {\bf E}\Big\{\Theta\Big|Y^{t}, V_0=v_0\Big\}, \hst  \Sigma_t\tri {\bf E}\Big\{ \Big(\Theta - \widehat{\Theta}_t\Big)^2 \Big|Y^t, V_0=v_0 \Big\}, \hso t=1, \ldots, n.
\end{align}
Consider the  linear coding scheme
\begin{align}
X_t=& g_t\Big(\Theta - \widehat{\Theta}_{t-1}\Big), \hso t=2, \ldots, n, \hso X_1=g_1 \Theta, \label{inn_b_a}  \\
Y_t =& g_t\Big(\Theta - \widehat{\Theta}_{t-1}\Big) +c_tV_{t-1} +W_t,   \hso t=2, \ldots, n,   \\
=& g_t\Big(\Theta - \widehat{\Theta}_{t-1}\Big) -c_tg_{t-1}\Big(\Theta - \widehat{\Theta}_{t-2}\Big) +c_t Y_{t-1}+ W_t  \hst \mbox{by $V_{t-1}=Y_{t-1}-X_{t-1}$} \label{inn_b_b}\\
Y_1=& g_1\Theta  +c_0 V_0+W_1, \\
{\bf E}&\big\{\sum_{t=1}^n \big(X_t\big)^2\Big|V_0=v_0\big\}= \frac{1}{n} \sum_{t=1}^n \kappa_t \leq \kappa,  \hst  \kappa_t =g_t^2 \Sigma_{t-1}, \hso t=2, \ldots, n,\hst \kappa_1=g_1^2 K_{\Theta}.  \label{inn_b_c}
\end{align}
Then the  statements of Theorem~\ref{thm_batman_gen}.(a), (b) hold, with the following changes:  all conditional expectations are replaced by conditional expectations for a fixed $V_0=v_0$,  $I(\Theta; Y^n)$ is replaced by $I(\Theta; Y^n|v_0)= H(Y^n|v_0)-H(V^n|v_0)$, and $\Sigma_1$ is replaced by $\Sigma_1=\frac{K_\Theta K_{W_1}}{g_1^2 K_\Theta +K_{W_1}}$.
\end{corollary}
 \begin{proof}
 This follows by repeating the derivation of Theorem~\ref{thm_batman_gen}.
 \end{proof}

From Theorem~\ref{thm_batman_gen} follows the validity of      Theorem~\ref{thm_bg_1}.

\begin{proposition}  Theorem~\ref{thm_bg_1} is correct.
\end{proposition}
\begin{proof} (a) By Theorem~\ref{thm_batman_gen}.(a) follows directly that optimization problem $C_n^L(\kappa, {\bf P}_{V_1})$ is as stated. \\
(b) This also follows from Theorem~\ref{thm_batman_gen}.(a), by replacing the total average power constraint by the pointwise average power constraint. \\
(c) This follows from Corollary~\ref{cor_batman_gen_a}, as in parts (a), (b).
\end{proof}

By Theorem~\ref{thm_batman_gen}, follows that $I(\Theta; Y^n)$ is maximized by the choice of the sequence $g_1, g_2, \ldots, g_n$ that controls the mean-square error $\Sigma_t, t=1, \ldots, n$, and satisfies the average power constraint. In general, this is a dynamic optimization problem, with state variable the sequence  $\Sigma_t, t=1, \ldots, n$.

From Theorem~\ref{thm_batman_gen}, we recover   Wolfowitz's \cite{wolfowitz1975} and Butman's \cite{butman1976}  lower bound, and more importantly the assumptions based on which the lower bound is derived. 

\begin{proposition} Reduction  of Theorem~\ref{thm_batman_gen}.(c) to  Butman's lower bound \\
\label{prop_but}
(a) The statements of  Theorem~\ref{thm_batman_gen}.(c), with $K_\Theta=1, K_{V_1}=K_W, c \in [-1,1]$, and strategy $g_n$ that satisfies,  (\ref{butman_cs__a2}), i.e.,  $sgn(g_n) =-sgn(c g_{n-1}), n=2,3, \ldots$, reduce to analogous statements derived by Butman \cite{butman1976}, i.e.,  (\ref{ihara_8})-(\ref{ihara_9}). \\
(a) If in addition to (a), i.e., $sgn(g_n) =-sgn(c g_{n-1}), n=2,3, \ldots$, the average power contraint is replaced by ${\bf E}\Big\{\big(X_t\big)^2\Big\} = \kappa$, $ t=1,2, \ldots, n$ (as in Butman \cite{butman1976}) then  problem  (\ref{ihara_8_aa})-(\ref{ihara_8_aaa}) is obtained.  
\end{proposition}
\begin{proof}
This is easily verified.
\end{proof}

\bibliographystyle{IEEEtran}
\bibliography{Bibliography_capacity}



\end{document}